\documentclass[a4paper,english,numberwithinsect, cleveref]{lipics-v2021}
\hideLIPIcs
\usepackage{microtype} % if unwanted, comment out

\title{Geometric Give and Take\footnote{%Supported by
J.T.\ supported by Charles Univ.\ projects UNCE 24/SCI/008 and PRIMUS 24/SCI/012, and by GAČR project 25-17221S.
}}
\titlerunning{Geometric Give and Take}

\author{Oswin Aichholzer}{Graz University of Technology}
  {oswin.aichholzer@tugraz.at}{}{}
\author{Katharina Klost}{ Freie Universit\"{a}t Berlin and Universität Tübingen}{kathklost@inf.fu-berlin.de}{}{}
\author{Kristin Knorr}{ Freie Universit\"{a}t Berlin}{knorrkri@inf.fu-berlin.de}{}{}
\author{Viola M\'{e}sz\'{a}ros}{Bolyai Institute, University of Szeged}
  {meszaros.viola@gmail.com}{}{}
\author{Josef Tkadlec}{Charles University}{josef.tkadlec@iuuk.mff.cuni.cz}{}{}
\authorrunning{O. Aichholzer, K. Klost, K. Knorr, V. M\'{e}sz\'{a}ros and J. Tkadlec}
\Copyright{Oswin Aichholzer, Katharina Klost, Kristin Knorr, Viola M\'{e}sz\'{a}ros and Josef Tkadlec}
\ArticleNo{100}
 %TODO mandatory: Please choose ACM 2012 classifications from https://dl.acm.org/ccs/ccs_flat.cfm 


\nolinenumbers

\usepackage{todonotes}
\usepackage{enumitem}

\usepackage{aliascnt}
\makeatletter
\let\corollary\@undefined
\let\endcorollary\@undefined
\let\lemma\@undefined
\let\endlemma\@undefined
\let\obs\@undefined
\let\endobs\@undefined
\let\claim\@undefined
\let\endclaim\@undefined
\makeatother
\theoremstyle{plain}
\newaliascnt{lemma}{theorem}
\newtheorem{lemma}[lemma]{Lemma}
\aliascntresetthe{lemma}

\newaliascnt{corollary}{theorem}
\newtheorem{corollary}[corollary]{Corollary}
\aliascntresetthe{corollary}

\newaliascnt{obs}{theorem}
\newtheorem{obs}[obs]{Observation}

\aliascntresetthe{obs}

\theoremstyle{claimstyle}
\newaliascnt{claim}{theorem}
\newtheorem{claim}[claim]{Claim}
\aliascntresetthe{claim}

\crefname{claim}{Claim}{Claims}
\Crefname{claim}{Claim}{Claims}

\begin{document}

\maketitle
\begin{abstract}
      We consider a special, geometric case of a balancing game introduced by Spencer in 1977. Consider any arrangement $\mathcal{L}$ of $n$ lines in the plane, and assume that each cell of the arrangement contains a box. Alice initially places pebbles in each box. In each subsequent step, Bob picks a line, and Alice must choose a side of that line, remove one pebble from each box on that side, and add one pebble to each box on the other side. Bob wins if any box ever becomes empty. We determine the minimum number $f(\mathcal L)$ of pebbles, computable in polynomial time, for which Alice can prevent Bob from ever winning, and we show that $f(\mathcal L)=\Theta(n^3)$ for any arrangement $\mathcal{L}$ of $n$ lines in general position.
\end{abstract}

\section{Introduction}
Inspired by the following game that appeared in the shortlist of problems for the International Math Olympiad in 2019 \cite{imo2019comittee60thInternationalMathematical2019}, we consider a special case of the balancing game introduced by Spencer~\cite{Spencer1977balancing} in 1977.
The game \cite{imo2019comittee60thInternationalMathematical2019} is
played with \(n\) boxes that are placed in a row and each box can contain a number of pebbles. 
There are two players: Alice and Bob. 
Initially, Alice places \(p_i\) pebbles into box \(i\).
Then the game proceeds in rounds: In each round, first Bob picks a line \(\ell_i\) with \(1\leq i < n\) that separates boxes \(1,\dots, i\) from boxes \(i+1,\dots, n\).
Second, Alice chooses one side of the line, removes one pebble from each box on that side, and adds one pebble to each box on the opposite side. %other side of the line.
(So, if Alice picked the left side of line $\ell_i$, then she removes one pebble from each box \(1,\dots, i\) and adds one pebble to boxes \(i+1,\dots, n\) each.)
The game ends with a win for Bob once any of the boxes becomes empty.
Alice's goal is to prevent Bob from winning.
The problem asks to determine the smallest total number of pebbles that Alice initially needs to put in the boxes, so that she can prevent Bob from winning.
Denoting this number by \(f_{1D}(n)\), it can be shown that \(f_{1D}(n)=\frac14n(n+4)\) if $n$ is even and \(f_{1D}(n)=\frac14(n+1)(n+3)-1\) if $n$ is odd, see~\cite{imo2019comittee60thInternationalMathematical2019}. In particular, \(f_{1D}(n)=\Theta(n^2)\).

In our work, we generalize the game to a 2-dimensional setting as follows:
Let \(\mathcal{L} = \{\ell_1,\dots, \ell_n\}\) be a set of $n$ lines in the plane.
Those lines split the plane into up to $\binom{n+1}{2}+1$ regions.
Each region contains one box.
As in the 1-dimensional problem, Alice first places pebbles in the boxes. The game then proceeds in rounds with Bob always picking a line from~\(\mathcal{L}\), and Alice picking a side, removing one pebble from each box on that side and adding a pebble to each box on the other side, see \cref{fig:one_round} for an example.
Given a line arrangement $\mathcal{L}$, let \(f(\mathcal{L})\) be the smallest number of pebbles such that Alice can distribute them to the boxes and prevent Bob from winning. The 1-dimensional setting corresponds to the case in which $\mathcal{L}$ consists of $n-1$ parallel lines.

\begin{figure}
    \centering
    \includegraphics[width=\linewidth]{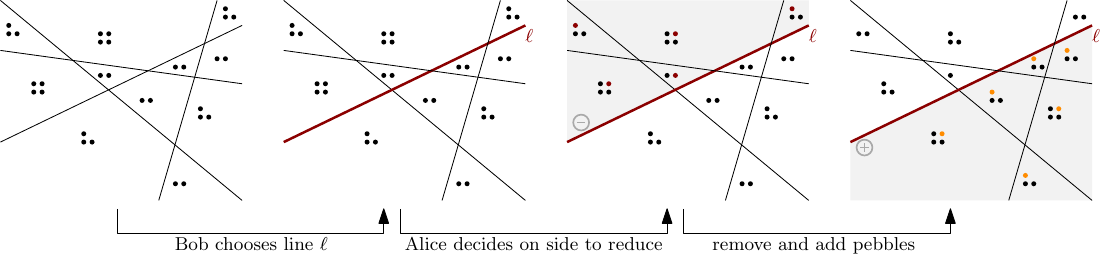}
    \caption{Example for one round of the 2-dimensional game}
    \label{fig:one_round}
\end{figure}

As our main result, we determine %the value
\(f(\mathcal{L})\) for any line arrangement $\mathcal{L}$.
A key notion 
here 
is measuring how (un)evenly the regions in the line arrangement are split by the respective lines.
More formally, given a line $\ell$ in a line arrangement~$\mathcal{L}$, we define the \textit{rank} of $\ell$, denoted $c_\ell$, to be the number of regions on the side of $\ell$ 
that contains
at most half of all the regions.

\begin{theorem}\label{thm:formula} Let $\mathcal{L}$ be a line arrangement that splits the plane into $R$ regions. Then $f(\mathcal{L})=R+\sum_{\ell\in\mathcal{L}} c_\ell$.
\end{theorem}

We also show that \(f(\mathcal{L})\) is cubic, provided that the line arrangement is in general position.

\begin{theorem}\label{thm:cubic}
Let $\mathcal{L}$ be a line arrangement with $n$ lines in general position. Then $f(\mathcal{L})\in\Theta(n^3)$.
\end{theorem}

Our results are organized as follows.
In \cref{sec:autopilot} we introduce the notion of autopilot strategies for Alice, and we prove that one of those autopilot strategies guarantees a win for Alice, provided that she is allowed to initially distribute at least $R+ \sum_{\ell\in\mathcal{L}} c_\ell$ pebbles. This establishes the ``$\le$'' inequality in~\cref{thm:formula}.

In \cref{sec:bob} we prove that with fewer than $R+ \sum_{\ell\in\mathcal{L}} c_\ell$ pebbles, Bob can force a win in a finite number of moves. In high level terms, Bob's strategy proceeds in stages, where in each stage either the total number of pebbles decreases, or the pebbles are shuffled from ``more important'' boxes to the ``less important'' boxes. Using a suitable monovariant, we show that eventually one box becomes empty. This establishes the ``$\ge$'' inequality in~\cref{thm:formula}.%, thus completing its proof.

Finally, in~\cref{sec:cubic} we prove that as long as the arrangement of $n$ lines is in general position, the quantity $\sum_{\ell\in\mathcal{L}} c_\ell$ is cubic in $n$. The argument uses the $\le k$-level Theorem~\cite{alon1986number}.

\subsection*{Related work}
The game considered in this work is a special case of a so-called \textit{balancing game} introduced in full generality by Spencer~\cite{Spencer1977balancing}.
There, the game is determined by a convex body $K\in \mathbb{R}^d$, a set $V\subseteq \mathbb{R}^d$ of $d$-dimensional vectors, and a starting position $\xi\in K$. In the $i$-th round, one player (the ``pusher'') picks a vector $v\in V$ and the other player (the ``chooser'') moves from the current position $\xi_{i-1}$ to either $\xi_i=\xi_{i-1}+v$ or to $\xi_i=\xi_{i-1}-v$. The pusher wins if the current position ever leaves $K$, and the chooser wins otherwise.

To see why the game considered here is a special case of a balancing game, fix any order over the $d\le \binom{n+1}{2}+1$ boxes, and represent the current configuration $(x_1,\dots,x_d)$ of pebbles as a point $(x_1,\dots,x_d)\in\mathbb{R}^d$.
Each line $\ell_i\in\mathcal{L}$ gives rise to a single vector $v_i\in\{- 1,+1\}^d$, and the convex region $K=K^+$ is the $d$-dimensional ``quadrant'' with all coordinates positive.

The balancing game was studied in several other works, see e.g.\ \cite{barany1979class} and references in \cite{fraenkel2012combinatorial}.
To our knowledge, the question of finding the smallest $\sum_i x_i$ for which the chooser wins with $K=K^+$ has not been studied before.

\section{Autopilot strategy for Alice}\label{sec:autopilot}
In this section we describe the \emph{autopilot strategy} for Alice that, given a line arrangement $\mathcal{L}$, allows her to initially place pebbles, and then play so as to prevent Bob from ever emptying any box. 
The main idea is that for every line in the arrangement, Alice first assigns a side where she will remove the pebbles when the line is first played. After that for every line she alternates the side from which she removes the pebbles. At the beginning, she assigns just enough pebbles to each box to guarantee that if every line is played exactly once, regardless of the order in which they are played no box becomes empty. 
We show that when Alice follows this strategy and updates the sides of the lines accordingly, the distribution of the pebbles will always maintain this property and thus Bob cannot win.
The upper bound on \(f(\mathcal{L})\) then follows from taking the assignment of the sides that leads to the smallest number of pebbles.

Formally, let \(\mathcal{L}=\{\ell_1,\dots, \ell_n\}\) be a line arrangement and let \(B\) be the set of boxes defined by~\(\mathcal{L}\). We call a function \(p: B \rightarrow  \mathbb{N}\) that assigns pebbles to the boxes a \emph{pebble distribution}. Consider the two half-planes defined by \(\ell\) and let \(\sigma(\ell)\) be a function that assigns an orientation "$+$" to one of the half planes, and the orientation "$-$" to the other. Let \(\tau_\sigma\colon B \rightarrow 2^\mathcal{L}\) be the function that assigns to each box \(b\) the set of all lines for which \(b\) lies in the side labeled  "$-$".

Initially, Alice chooses a certain orientation \(\sigma\) and for each box \(b\) she sets \begin{equation}
    p(b) = |\tau_\sigma(b)| + 1. \label{inv:autopilot}
\end{equation}
We call  pebble distributions defined with respect to an orientation of the lines \emph{autopilot distributions}.
During the game, Alice will update the orientation \(\sigma\), such that (\ref{inv:autopilot}) always holds.
Let \(\sigma_t\) be the orientation after \(t\) rounds and let \(p_t\) be the corresponding autopilot distribution.
Let \(\ell\) be the line picked by Bob in round \(t\). Then Alice removes one pebble from each box in the side of \(\ell\) labeled "$-$" and adds one pebble to each box on the other side. She then gets \(\sigma_{t+1}\) by flipping  the labels of the two sides of \(\ell\).
We call this the \emph{autopilot strategy}.

\begin{lemma}\label{lem:autopilot_inv}
If Alice starts with any autopilot distribution and plays the autopilot strategy, all distributions \(p_0,\dots \) of the pebbles to the boxes will be autopilot distributions. 

\end{lemma}
\begin{proof}%
We show this by induction over the number of rounds. 
The base case is clear by the statement. 
Now consider the changes in round \(t\). The orientation of exactly one line \(\ell\) differs between \(\sigma_t\) and \(\sigma_{t+1}\). Thus we have to argue, that \(p_{t+1}(b)\) changes according to (\ref{inv:autopilot}).
Let \(b\) be a box such that \(\ell \in \tau_{\sigma_t} (b)\). 
Then \(\ell \notin \tau_{\sigma_{t+1}}(b)\) and similarly, if \(\ell\notin \tau_{\sigma_t} (b)\), then  \(\ell\in \tau_{\sigma_{t+1}} (b)\).
It directly follows that removing a pebble from all boxes of the first kind and adding a pebble in each box of the second kind maintains the invariant.
\end{proof}

Since every autopilot distribution has at least one pebble in each box, we conclude: %have the following Corollary.
\begin{corollary}\label{cor:autopilot_win}
When Alice starts with an autopilot distribution and plays the autopilot strategy, Bob has no winning strategy.
\end{corollary}

As a direct consequence, we obtain the ``$\le$'' inequality in \cref{thm:formula}.
\begin{lemma}\label{lem:formula-UB}
Let $\mathcal{L}$ be any line arrangement. Then $f(\mathcal{L})\le R + \sum_{\ell\in\mathcal{L}} c_\ell$.
\end{lemma}
\begin{proof} For the initial distribution of pebbles, Alice chooses the autopilot distribution \(p_a\) that, for each line $\ell\in\mathcal{L}$, assigns the orientation ``$-$'' to the half-plane that has at most half of the boxes. In this way, Alice has distributed precisely $\sum_{b\in B} p_a(b) = R + \sum_{\ell\in\mathcal{L}} c_\ell$ pebbles. If she then follows the autopilot strategy, by \cref{cor:autopilot_win} no box will ever become empty.
\end{proof}

\section{Monovariant strategy for Bob}\label{sec:bob}
In this section we give a strategy for Bob that allows him to win, if Alices places fewer than \(R+\sum_{\ell\in\mathcal{L}} c_\ell\) pebbles.  
The game is subdivided into stages and we can show that after every stage either Bob wins, or the overall number of pebbles decreased, or the pebble distribution is smaller in some fixed linear order of the pebble distributions. 
As the latter two cases can only happen finitely many times, Bob will eventually win.

\begin{lemma}\label{lem:formula-LB} Let $\mathcal{L}$ be any line arrangement. Then $f(\mathcal{L})\ge R + \sum_{\ell\in\mathcal{L}} c_\ell$.
\end{lemma}
\begin{proof}
Fix $\mathcal{L}$. 
We need to prove that with fewer than $X= R +\sum_{\ell\in\mathcal{L}} c_\ell$ pebbles in total, Bob can force a win.
Label the boxes by $B=\{b_1,b_2,\dots,b_R\}$ in any order and let \(p\) be the current pebble distribution. Furthermore, let \(P=\sum_{b\in B}p(b) <X\) be the total number of pebbles.

For each line $\ell\in\mathcal{L}$ we define its \textit{smaller} side as the side that either contains fewer than half of the boxes, or it contains precisely half of the boxes and does \textit{not} contain the box~$b_1$. 
Consider the autopilot distribution obtained by assigning the orientation ``$-$'' to the smaller side of each line, and denote by $p_a(b_i)$ the number of pebbles in box $b_i$ in this autopilot distribution. Note that $\sum_{b_i\in B} p_a(b_i)=X>P$.

We also define a total order on the set of all pebble distributions, by considering the lexicographic order on the lists \([p(b_1), p(b_2),\dots, p(b_R)]\). We write \(p \prec_L p'\) if \(p\) is before \(p'\) in the lexicographic ordering.

Bob's strategy proceeds in stages. Let \(p\) be the pebble distribution before a stage and \(p'\) be the pebble distribution at the end of the stage. We show that after each stage one of the following holds:
\begin{enumerate}
    \item Bob wins, that is, there is a box \(b\) with \(p'(b)=0\)
    \item The total number of pebbles decreases or \(p'\prec_L p\).
    \end{enumerate}
    Since the second situation can only happen finitely many times, Bob will win.

At the beginning of each stage, we have $P<X$.
Thus, by the pigeonhole principle, there exists a box $\textbf{b}$ that has strictly fewer than $p_a(\textbf{b})$ pebbles (that is, $p(\textbf{b})\le p_a(\textbf{b})-1$).
Call this box the \textit{focal box}.
Next, Bob will consider all the $A=p_a(\textbf{b})-1$ lines $\ell$ that have the focal box $\textbf{b}$ on their smaller side and he will sort those lines as follows. First he sorts them into \textit{batches} by decreasing rank $c_\ell$.  Within each batch he sorts lexicographically decreasing by the characteristic vector of the boxes on the smaller side of \(\ell\). That is, by \(v^\ell=(v^\ell_1,\dots, v^\ell_R) \in \{0,1\}^R\) with \(v^\ell_i=1\) if and only if \(b_i\) is on the smaller side of \(\ell\). 
See \cref{fig:batch_sorting} for an illustration.

\begin{figure}
    \centering
    \includegraphics{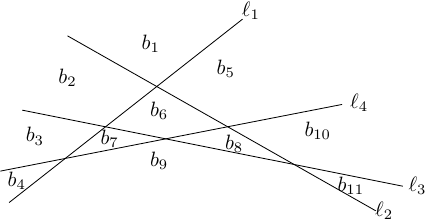}
    \caption{The indices of the lines show how all the lines would be ordered, if \(\textbf{b}\) was on the smaller side of all of them (each line has 4 or 5 boxes on their smaller side). \(\ell_1\) is before \(\ell_2\) as the first position of \(v^{\ell_1}_1 = v^{\ell_2}_1\) agree, but \(v^{\ell_1}_2 = 1\) and \(v^{\ell_2}_2=0\).
    If \(\textbf{b}=b_9\) Bob only plays \(\ell_3\) and \(\ell_4\)}
    \label{fig:batch_sorting}
\end{figure}

Denote the lines in this sequence by $S=(\ell_1,\dots,\ell_A)$. 
We will now focus on the smaller side of each line and say that Alice \emph{decreases} a line if she removes the pebbles from that side and otherwise she \emph{increases} the line.

Bob's strategy within a stage is as follows:
Bob starts by playing $\ell_1$.
In general, suppose that Bob has just played $\ell_i$. If Alice decreases \(\ell_i\), then Bob will next play $\ell_{i+1}$. If Alice increases, then Bob will next play $\ell_{i-1}$.
When Alice decreases \(\ell_A\), increases \(\ell_1\) or a box becomes empty, the stage ends.

Another way to describe Bob's strategy is by considering a stack. Whenever Alice decreases a line \(\ell_i\), it is pushed to the stack and Bob next plays \(\ell_{i+1}\). If Alice increases a line \(\ell_j\) and the stack is not empty, \(\ell_{j-1}\) is the topmost line on the stack. Then \(\ell_{j-1}\) is popped from the stack and paired with \(\ell_j\). Bob next plays line \(\ell_{j-1}\).
If the stack is empty on an increase, \(\ell_A\) is pushed to the stack or a box becomes empty, the stage ends.

This decomposes the sequence of lines that Bob plays into the set of pairs and the content of the stack. The content of the stack is a, possibly empty, sequence \(\ell_1,\dots, \ell_k\).
We say that a pair \((\ell_i, \ell_{i+1})\) is:
\begin{enumerate}
\item \textit{red} if the rank of $\ell_i$ is larger than the rank of $\ell_{i+1}$, that is, they are in different batches.

\item \textit{green} if the ranks of $\ell_i$ and $\ell_{i+1}$ are the same, that is, they are in the same batch.

\end{enumerate}

See~\cref{fig:convex-c} for an illustration. 

\begin{claim}\label{claim:pairs}
The pairs have the following effects on the pebble distribution.
\begin{enumerate}[label=(\roman*)]
    \item No pair changes the number of pebbles in \(\textbf{b}\).   \label{item:bfunchanged}
    \item Each red pair decreases the overall number of pebbles by at least 2. \label{item:redDecrease}
    \item Each green pair leaves the overall number of pebbles unchanged. \label{item:greenUnchanged}
    \item If \(p\) and \(p'\) are the pebble distributions before and after a green pair respectively, then \(p' \prec_L p\). \label{item:greenLexDecrease}
\end{enumerate}
\end{claim}
\begin{claimproof}
Let \((\ell_i, \ell_{i+1})\) be a pair. Recall that \(\ell_i\) was decreased and \(\ell_{i+1}\) was increased.
    \begin{enumerate}[label=(\roman*)]
        \item For any pair, decreasing \(\ell_i\) removes a pebble from \(\mathbf{b}\) and increasing \(\ell_{i+1}\) adds a pebble to \(\mathbf{b}\), leaving \(p(\mathbf{b})\) unchanged.
        \item
        Decreasing \(\ell_i\) changes the overall number of pebbles by \((R - c_{\ell_i}) - c_{\ell_i}\). Increasing \(\ell_{i+1}\) changes the overall number of pebbles by \(c_{\ell_{i+1}} - (R-c_{\ell_{i+1}})\). Thus the overall change is 
        \begin{equation*}
           R-2c_{\ell_i} + 2c_{\ell_{i+1}} - R = 2 (c_{\ell_{i+1}} - c_{\ell_{i}})
        \end{equation*}
        As the pair is red by assumption, we have \(c_{\ell_i}> c_{\ell_{i+1}}\) and the statement follows.
        \item This follows from the same argument as (ii), using that \(c_{\ell_i} = c_{\ell_{i+1}}\).
        \item Consider the pebble distributions \(p\) and \(p'\) before and after the green pair respectively  and let $B_i$, $B_{i+1}$ be the sets of boxes on the smaller sides determined by $\ell_i$, $\ell_{i+1}$, respectively. 
The net effect of decreasing $\ell_i$ and increasing $\ell_{i+1}$ is that each box $b\in B_i\setminus B_{i+1}$ loses two pebbles, and each box $b'\in B_{i+1}\setminus B_i$ gains two pebbles. Since the lines $\ell_i$, $\ell_{i+1}$ have the same rank (and are different), we have $|B_i\setminus B_{i+1}|=|B_{i+1}\setminus B_i|>0$, and thanks to the ordering by characteristic vectors we get \(p'\prec_L p\). \claimqedhere
    \end{enumerate}
\end{claimproof}
 \begin{figure}[h!]
    \centering
    \includegraphics[width=0.8\linewidth]{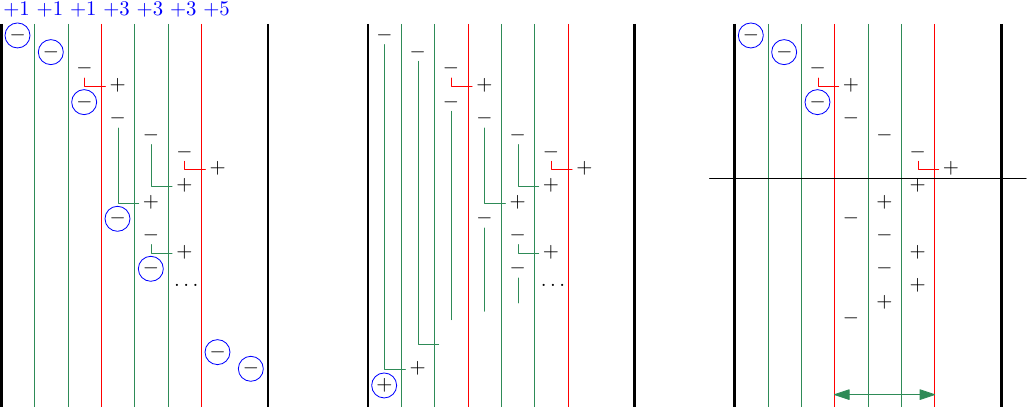}
    \caption{Illustrations for the three cases which correspond to (1) Alice decreasing $\ell_A$ (left); (2) Alice increasing $\ell_1$ (middle); and (3) Alice never decreasing $\ell_A$ and never increasing $\ell_1$ (right).}
    \label{fig:convex-c}
\end{figure}

Next, we argue that Bob wins in a finite number of steps. 
We distinguish 3 cases, depending on what happens in a stage.

    \subparagraph*{(1) Alice decreases \(\ell_A\)} %We claim that then the focal box \(\textbf{b}\) is empty. 
    Let \(p\) be the initial pebble distribution of the stage.
    Each decrease of a line that is still on the stack removes one pebble from \(\textbf{b}\). Furthermore, by \cref{claim:pairs} \ref{item:bfunchanged} no pair changes the number of pebbles in \(\textbf{b}\).
    At the beginning of the stage there were \(p(\textbf{b})\leq p_a(\textbf{b})-1=A\) pebbles in \(\textbf{b}\) and each decrease on the \(A\) lines on the stack removed a pebble. Thus, \(\textbf{b}\) is now empty and Bob wins.
    \subparagraph*{(2) Alice increases \(\ell_1\)} Let \(p\) be the initial pebble distribution of this stage and \(p’\) be the pebble distribution after the stage. We show that in this case, the total number of pebbles decreased or \(p'\prec_L p\).
    By \cref{claim:pairs} \ref{item:redDecrease} and \ref{item:greenUnchanged}, red pairs decrease the total number of pebbles, green pairs do not change it. After Alice increases \(\ell_1\), the stack is empty, and thus the number of pebbles did not increase in this stage. 
    Thus, if there is a red pair or  $\ell_1$ has rank strictly less than $R/2$, then increasing \(\ell_1\) decreases the total number of pebbles.

In the other case, all pairs are green \textit{and} $\ell_1$ has order exactly $R/2$. In that case, the effect of increasing $\ell_1$ and by \cref{claim:pairs} \ref{item:greenLexDecrease} the effect of any green pair is to decrease the pebble distribution with regard to \(\prec_L\).
In particular, \(b_1\) is on the larger side of \(\ell_1\) by definition and thus increasing \(\ell_1\) decreases the number of pebbles in the most significant position. If follows that \(p'\prec_L p\) as claimed.

As the total number of pebbles is still smaller than \(X\), Bob can ignore all previous moves, pick a new focal box and start a new stage.
Since the number of pebbles and the lexicographical ordering of the pebble distribution can only decrease a finite number of times, there are only finitely many stages that end with Alice increasing \(\ell_1\).

\subparagraph*{(3) Alice never decreases \(\ell_A\) or increases \(\ell_1\)} We show that in this case, the stage ends with reaching some empty box.
By \cref{claim:pairs} \ref{item:bfunchanged}, the pairs do not increase the total number of pebbles.
The lines on the stack increase the initial number of pebbles \(P\)
by at most  $Z=n\cdot R\le n\cdot \left(\binom{n}{2}+1\right)$. Since by \cref{claim:pairs} \ref{item:redDecrease} each red pair removes at least 2 pebbles in total, the sequence may contain at most $\frac12(P+Z)$ red pairs.

Consider the moment of the last occurrence of a red pair (or the very beginning, if no red pair ever occurs). 
In the subsequence following this moment, all pairs are green and all the lines involved in them have the same rank.

Let $(\ell_i,\ell_{i+1})$ be any green pair and  \(p, p'\) be the pebble distributions before and after the green pair. Then by \cref{claim:pairs} \ref{item:greenLexDecrease} after each green pair we have \(p' \prec_L p\).
Since the total number of pebbles is finite, after finitely many moves some box becomes empty and Bob wins. \qedhere
\end{proof}

\begin{proof}[Proof of \cref{thm:formula}]
    Immediate from \cref{lem:formula-LB,lem:formula-UB}.
\end{proof}

\section{General position}
\label{sec:cubic}
In this section we prove \cref{thm:cubic}.
The upper bound turns out to be simple (any autopilot will do). For the lower bound, we introduce several auxiliary notions, we prove a helpful lemma, and we use a known result about $\le k$-levels.

\begin{definition}
Let $\mathcal L$ be a line arrangement. Its \textit{dual graph} is the unweighted graph $D(\mathcal L)=(B,E)$ that has one vertex for each box, 
and edges between boxes that share a side.

The distance of two nodes $u$, $v$ in $D(\mathcal L)$ is denoted $d(u,v)$.
Moreover, given a pebble distribution $p\colon B\to \mathbb{N}$, we say that a pair $u,v\in B$ is \textit{small} if \(p(u)+p(v) < d(u,v) + 2\).
\end{definition}

\begin{lemma}\label{lem:local_bob1}
    If a pebble distribution $p$ contains a small pair, then Bob has a winning strategy.
\end{lemma}
\begin{proof}
    This proof is a generalization of known facts in the \(1d\) setting~\cite{imo2019comittee60thInternationalMathematical2019}. As these solutions have not been formally published, we include a proof for completeness.

    Let \(u,v\) be the small pair and \(\pi(u,v) = (u=b_{1},\dots, b_{k}=v) \) be a shortest path that connects \(u\) and \(v\). 
    Clearly, \(\pi(u,v)\) crosses each line that separates \(u\) from \(v\) at least once. 
    On the other hand, there is path defined by a line segment between any fixed point in \(u\) and any fixed point in \(v\), and this segment crosses every line that separates \(u\) from \(v\) exactly once.
    Both observations together imply that the edges in \(D(\mathcal{L})\) in \(\pi(u,v)\) are each defined by a different line of~\(\mathcal{L}\). Without loss of generality, let these lines be \(\ell_{1},\dots, \ell_{k-1}\).

    Let \(p_t(u), p_t(v)\) be the number of pebbles in \(u\) and \(v\), respectively, after \(t\) rounds of the game. 
    Bob will now restrict the lines he plays to \(\ell_1,\dots, \ell_{k-1}\). This directly implies that in each move of Alice, one pebble is added to either \(u\) or \(v\) and one pebble is removed from the other box.
    This gives the invariant for all rounds \(t, t'\):
    \begin{equation}
        p_{t}(u) + p_{t}(v) = p_{t'}(u) + p_{t'}(v)\label{inv:local_bob1}
    \end{equation}
    Bobs strategy is to play line \(\ell_{p_{t-1}(u)}\) in round \(t\).
    Note, that this implies, that \(p_t(u) \in \{p_{t-1}(u) -1, p_{t-1}(u) +1\}\).

    We show that Bob has a winning strategy by induction over the number of times where Alice switches between subtracting from all boxes on the same side as \(u\) and all boxes on the same side of~\(v\).
    For the base case, first assume that Alice always subtracts from the side that contains \(u\).
        Then after \(t\) rounds we have  \(p_t(u) = p_0(u)-t\) and Bob played line \(\ell_{p_0(u)-t}\).
        In round \(t^*=p_0(u)-1\), Bob plays \(\ell_1\) and after this round it holds that \(
        p_{t^*}(u) = p_0(u) - p_0(u)=0\) and Bob wins.
        
        If Alice always picks the side that contains \(v\), then we have \(p_t(u) = p_0(u)+t\). In round \(t^* = p_0(v)\), Bob plays line \(\ell_{p_0(u)+p_0(v)-1}\) and since \(u,v\) are a small pair, this is one of the lines defined by \(\pi(u,v)\). Furthermore, we have \(p_{t^*}(u) = p_0(u) + p_0(v)\) and by (\ref{inv:local_bob1}) this implies that \(p_{t^*}(v) = 0\), again Bob wins. 

        For the inductive step assume, without loss of generality, that up to round \(t\) Alice always took pebbles away from the side with box \(u\) and at round \(t+1\) she takes away from box \(v\).
        Let \(\ell_{i}\) be the line played in step \(t\). Then Bob plays line \(\ell_{i-1}\) in step \(t+1\).
        These two steps give \(p_{t-1}(b_j) = p_{t+1}(b_j)\) for  \(j\neq i\) since each of these boxes is once on the side where the pebbles are subtracted and once on the side where they are added. 
        However, we have \(p_{t+1}(b_i) = p_t(b_i)-1 = p_{t-1}(b_i)-2\).
        So Bob is in the same situation as two rounds before, that is, he will next play line $\ell_{i}$. But in total there are two pebbles less in the boxes along the path $\pi(u,v)$ than two rounds before. 
        In particular, \(u,v\) are still a too-small pair and by the inductive hypothesis, Bob wins.
\end{proof}

The second ingredient is a version of a $\le k$-level Theorem that can be stated for sets of points, or dually for the line arrangements, see also a review~\cite[Section 1.3]{wagner2008k}.

\begin{theorem}[Theorem 1 in~\cite{alon1986number}]\label{thm:k-level}
    Let $\mathcal L$ be an arrangement of $n$ lines, $D(\mathcal L)=(V,E)$ its dual graph, and $1\le r<n/2$. Then for any $u\in V$ we have
    \(| \{ v\in V\mid d(u,v)\le r \}| \le r\cdot n. \)

\end{theorem}

Finally we are ready to prove \cref{thm:cubic}.

\begin{proof}[Proof of \cref{thm:cubic}]
First, note that any autopilot strategy uses $\mathcal{O}(n^3)$ pebbles: Indeed, in every autopilot distribution, all boxes receive at most $n+1$ pebbles. Since for a line arrangement in general position the number of boxes is $\binom{n+1}{2}+1$, Alice can distribute $\mathcal{O}(n^3)$ pebbles and prevent Bob from winning.

Next, we prove the lower bound.
Let $p$ be any pebble configuration such that Alice wins.
If each box has at least $\frac1{10}n$ pebbles then the total number of pebbles is at least $\frac1{10}n\cdot \frac12n^2 =\Omega(n^3)$ and we are done.
Suppose otherwise and let $b\in B$ be a box with fewer than $\frac1{10}n$ pebbles.
Fix a distance $r=\frac13n$.
By~\cref{thm:k-level} there are at most $r\cdot n=\frac13n^2$ boxes at distance at most $r$ from $b$, so there are at least $\frac12n^2-\frac13n^2=\frac16n^2$ boxes at distance at least $\frac13n$ from $b$. 
By~\cref{lem:local_bob1}, each of those boxes must have at least $r-n/10\ge \frac{7}{30}n$ pebbles, else Bob wins. 
Thus, the total number of pebbles is at least $\frac16 n^2\cdot \frac7{30}n = \Omega(n^3)$.
\end{proof}

\section{Discussion}
\begin{figure}
    \centering
    \includegraphics[width=0.8\linewidth]{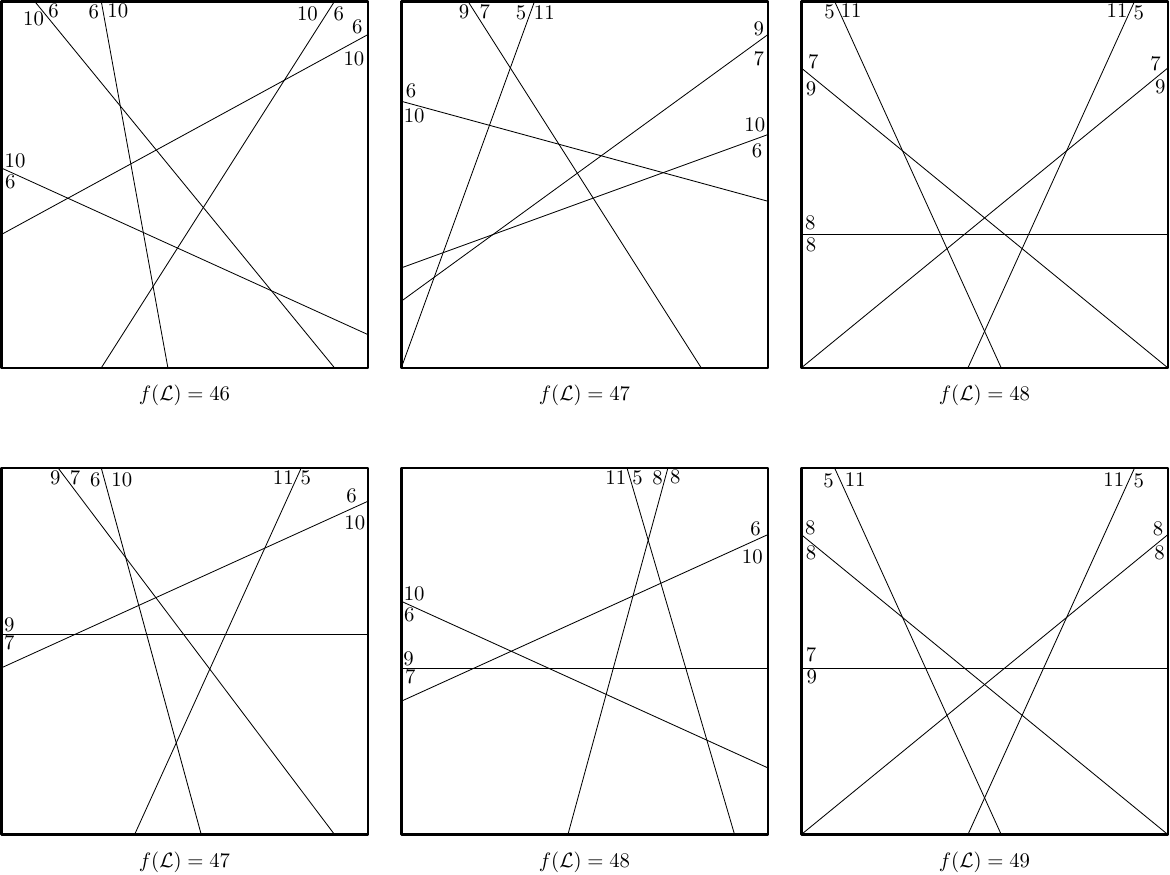}
    \caption{The six combinatorially different simple line arrangements for $n=5$ lines~\cite{simple_arrangements}. Shown is the number of cells on both sides of each line. The optimal autopilot for each arrangement requires $46\le f(\mathcal{L})=\sum_{\ell\in\mathcal{L}} c_\ell + 16\le 49$ pebbles.} %\pepa{I checked this.}
    \label{fig:n5}
\end{figure}

There are many directions for future research, we briefly mention two of them:

First, given an integer $n$, determine the smallest and the largest possible value of $f(\mathcal L_n)$, where $\mathcal L_n$ is an arrangement of $n$ lines in general position.
For example, when $n=5$ a check of all combinatorially different simple line arrangements~\cite{simple_arrangements} shows that $46\le f(\mathcal{L})\le 49$, see~\cref{fig:n5}.

Second, is there a polynomial-time algorithm that, given a line arrangement and the pebbles in the boxes, decides whether Bob has a winning strategy? We believe we can show that the answer is yes if some box contains a single pebble, but we don't know about the general case.

 \subparagraph{Acknowledgment} We would like to thank the organizers of GGWeek 2025, where this research started.

\bibliography{GGT}

\end{document}